\def\ACKS{}
\documentclass[12pt]{article}
\usepackage{natbib}

\usepackage{geometry}
\usepackage{graphicx}
\graphicspath{{./graphics/}}

\usepackage{caption}

\usepackage[shortcuts]{extdash}
\usepackage{amsmath,amsfonts}

\usepackage{amsthm}
\theoremstyle{plain}
\newtheorem{theorem}{Theorem}
\newtheorem{lemma}[theorem]{Lemma}
\newtheorem*{theorem*}{Theorem}
\newtheorem{corollary}[theorem]{Corollary}
\theoremstyle{definition}
\newtheorem*{remark*}{Remark}

\newcommand{\range}[2]{\in\{#1,\dots,#2\}}
\newcommand{\allrooms}{\mathcal{M}}
\newcommand{\allagents}{\mathcal{N}}
\newcommand{\price}{\mathbf{p}}

\begin{document}

\title{Generalized Rental Harmony}

\author{Erel Segal-Halevi
%\\
%Ariel University
%\\
%Kiriat Hamada 3, Ariel 40700, Israel
%\\
%erelsgl@gmail.com
}
\date{}
\maketitle

\begin{abstract}
Rental Harmony is the problem of assigning rooms in a rented house to tenants with different preferences, and simultaneously splitting the rent among them, such that no tenant envies the bundle (room+price) given to another tenant.
Different papers have studied this problem  under two incompatible assumptions: the miserly tenants assumption is that each tenant prefers a free room to a non-free room; the quasilinear tenants assumption is that each tenant attributes a monetary value to each  room, and prefers a room of which the difference between value and price is maximum.
This note shows how to adapt the main technique used for rental harmony with miserly tenants, using Sperner's lemma, to a much more general class of preferences, that contains both miserly and quasilinear tenants as special cases.
This implies that some recent results derived for miserly tenants apply to this more general preference class too.
\end{abstract}


%\newpage
\section{The Problem}
\label{sec:intro}
There is a set $\allagents$ of agents who rent a house together.
The set of rooms in the house is $\allrooms$, with  $|\allagents|=|\allrooms|=n$. 
The total rent is $R$.
It is required to assign a price $p_j\in\mathbb{R}$ to each room $j\in\allrooms$ such that the sum of all prices is $R$,
and assign each room in $\allrooms$ to a unique agent in $\allagents$.
The agents have different preferences on the rooms. The preferences of an agent are represented by a \emph{demand function} --- for each price-vector $\price\in\mathbb{R}^n$, the agent has a set of one or more rooms that he/she considers the \emph{best rooms} given the prices.
An allocation in which 
each agent receives a room from the set of his/her best rooms is called \emph{envy-free}.

The existence of envy-free allocations has been proved using various techniques and under various assumptions on the agents' preferences. \citet{gale1989theory} made the following assumption:
\begin{quote}
\textbf{Quasilinear Tenants Assumption:}
For all  $i\in\allagents$ and $j\in \allrooms$, 
there is a value $v_{i,j}$ representing the value for agent $i$ of room $j$. The best rooms of agent $i$ in price $\price$ are the ones maximizing the difference $v_{i,j} - p_j$. 
\end{quote}
Gale proved the existence of envy-free allocations using linear programming duality. 
Later works provided fast algorithms for calculating an envy-free allocation with quasilinear tenants using various techniques: a compensation procedure \citep{Haake2002Bidding},
a market-based mechanism \citep{Abdulkadiroglu2004Room},
maximum matching and linear programming \citep{Sung2004Competitive}.
The latter approach has been later implemented and tested in the popular website spliddit.org \citep{gal2017fairest}.

\citet{Su1999Rental}, who invented the term ``rental harmony'', made a different assumption --- an agent always prefers a free room if one is available. Su considered only non-negative prices; since we will want to consider general prices, we present a slightly generalized version of his assumption --- an agent always prefers a room with a non-positive price if one is available:
\begin{quote}
\textbf{Miserly Tenants Assumption:}
Given a price-vector $\price$ in which $p_{j}\leq 0$ for some $j\in \allrooms$, every agent has a best room $j^*$ with 
$p_{j^*} \leq 0$.
\end{quote}
In addition, Su assumed that the demand functions are \emph{continuous} in the following sense: if some room $j$ is a best room for agent $i$ for a convergent sequence of price-vectors, then $j$ is a best room for $i$ in the limit price-vector.

Miserly tenants, in general, do not satisfy the quasilinear assumption. The miserly tenants model puts no restrictions on the preferences when all prices are strictly positive. The preferences may even depend on the entire price-vector. For example, the miserly tenants model allows an agent to prefer the first room when the price-vector is $(100,200,200,300)$ and prefer the second room when the price-vector is $(100,200,400,100)$, even though the prices of both rooms have not changed. \citet{Azrieli2014Rental} describe several situations in which such preferences may be reasonable.
For example, the agent may believe that a room with a rent as high as $400$ attracts wealthy neighbors and that this makes the neighboring room more attractive.
As another example, if the second and third room are similar in quality,
the fact that the third room costs $400$ might make the second room look more attractive --- a well-known behavioral bias.
In any case, attaining rental harmony with miserly tenants requires a different technique, which is described next.

\section{Sperner-type Lemmas}
\label{sec:sperner}
\begin{figure}
\includegraphics[height=6cm]{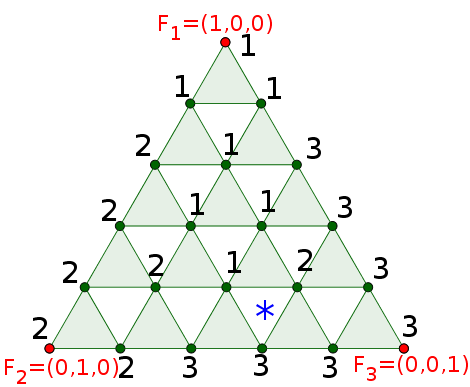}
\hskip 1cm
\includegraphics[height=6cm]{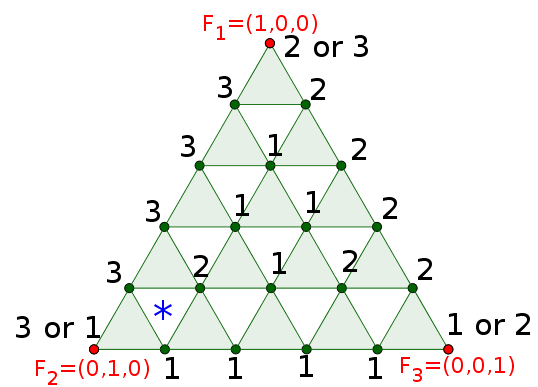}
\caption{
\label{fig:sperner}
Two labelings of a triangulation of a $3$-simplex. 
\textbf{Left}: the labeling satisfies the conditions to Sperner's lemma.
\textbf{Right}: the labeling satisfies the conditions to Scarf's lemma.
In both cases, * denotes a fully-labeled simplex (there may be more than one).
}
\end{figure}
Sperner's lemma and its variants consider an $n$-vertex simplex, which w.l.o.g. is the standard $n$-simplex 
$\Delta^{n-1} = \{(x_1,\ldots,x_n) | \sum_{j=1}^n x_j = 1, \forall j\in[n]: x_j\geq 0\}$.
For each $i\in [n]$, we denote by $F_i$ the main vertex of the simplex in which the $i$-th coordinate is $1$ and all other coordinates are $0$.

The simplex is triangulated, and each vertex of the triangulation is labeled with a label in $\{1,\ldots,n\}$.
The goal of all Sperner-type lemmas is to identify conditions that guarantee the existence of a \emph{fully-labeled sub-simplex} --- a sub-simplex of the triangulation whose vertices are labeled with $n$ different labels.

We say that a labeling satisfies
\emph{Sperner's boundary condition}
if the label on each vertex $(x_1,\ldots,x_n)$ is some $j$ for which $x_j > 0$.
In particular, the label of $F_1$ 
is $1$; the label of each vertex on the line between $F_1$ and $F_2$ is either $1$ or $2$; and so on.
Sperner's lemma says that \emph{every labeling that satisfies Sperner's boundary condition admits a fully-labeled simplex}. An example is shown in Figure \ref{fig:sperner}/Left.

Sperner's lemma has many variants; of particular interest here is the variant proved by \citet{scarf1982computation}. 
We say that a labeling satisfies
\emph{Scarf's boundary condition}
if the label on each vertex $(x_1,\ldots,x_n)$ is some $j$ for which $x_j = 0$. In particular, when $n=3$, the label of $F_1$ is either 2 or 3; the label of all vertices in the interior of the line between $F_1$ and $F_2$ is 3; and so on.
Scarf's lemma says that, in any sufficiently-fine triangulation (specifically, in any triangulation in which no sub-simplex touches all $n$ faces of the large simplex), \emph{every labeling that satisfies Scarf's boundary condition admits a fully-labeled simplex}. An example is shown in Figure \ref{fig:sperner}/Right.

\section{Miserly Tenants and Scarf's Lemma}
\label{sec:miserly}
Su's proof of rental harmony with miserly tenants has two components.
First, associate each vertex of the triangulation to one of the $n$ agents, such that in each sub-simplex, all $n$ agents are represented (this is easy to do with the triangulations illustrated in Figure \ref{fig:sperner}).
Then, associate each point $\mathbf{x}$ in $\Delta^{n-1}$
with a price-vector
$\mathbf{p}$ such that, for all $j\in\allrooms$,
\begin{align}
p_j = R\cdot x_j.
\end{align}
Each vertex of the triangulation is labeled with the index of one of the ``best rooms'' of the agent who owns that vertex, in the price associated with that vertex.

Each vertex on the boundary of $\Delta^{n-1}$ corresponds to a price-vector in which one or more rooms are free. Then, the Miserly Tenants assumption implies that for each agent there is a best room $j\in\allrooms$ for which $x_j=0$. Thus, it is possible to label each vertex of the triangulation such that the labeling satisfies 
Scarf's boundary condition. 
By Scarf's lemma, there is a fully-labeled simplex.
Consider a sequence of finer and finer triangulations. The sequence of fully-labeled simplices has a subsequence that converges to a point. 
At the limit point, by the continuity of the preferences, each agent has a different best room. Hence there is an envy-free allocation with  prices determined by the coordinates of the limit point.

Su's theorem has been extended in various ways: 
\begin{enumerate}
\item \citet{Azrieli2014Rental} considered \emph{rental harmony with roommates}, when each room can accommodate several tenants; 
\item \citet{frick2019achieving,asada2018fair} considered \emph{rental harmony with a secretive agent},
when only $n-1$ agents are present, and they have to determine a price-vector such that, when the $n$-th agent comes and picks a room he prefers, the other $n-1$ agents can allocate the remaining $n-1$ rooms among them in an envy-free way.
\item \citet{meunier2019multilabeled} considered \emph{rental harmony with an extra agent},
when $n+1$ agents are present, and they have to determine a price-vector such that, when any agent leaves, the remaining $n$ agents can allocate the $n$ rooms among them in an envy-free way.
\item 
\citet{nyman2020fair} considered \emph{multi-house rental harmony}, when there are several $n$-room houses (say, a bedroom house and an office building that are being rented to a set of agents together), and each agent should receive a room in each house.
\end{enumerate}
All these extensions use appropriate  generalizations of Scarf's lemma, and they all make an assumption similar to the Miserly Tenants assumption.

Quasilinear tenants do not satisfy the Miserly Tenants assumption. For example, consider a house with total rent $R=1000$ and three rooms: a spacious living-room and two basements. 
Consider a quasilinear tenant who values the living-room at $800$ and each basement at $100$. If the price-vector is $(600,400,0)$, then the quasilinear tenant (quite understandably) strictly prefers the living-room to both basements.
To further illustrate the difficulty with the miserly tenants assumption, note that this assumption combined with the continuity of the demand function implies that every agent is indifferent between all free rooms \citep{frick2019achieving}, which is clearly unrealistic.

\citet{Su1999Rental} notes, in the ``Comments and Discussion'' section, that his proof can be adapted to use a weaker assumption: whenever a free room is available, each agent has a best room that is not the most expensive room. Formally:
\begin{quote}
\textbf{Weak Miserly Tenants Assumption:}
Given a price-vector $\price$ in which $p_{j}\leq 0$ for some $j\in \allrooms$, every agent has a best room $j^*$ with 
$p_{j^*}  < \max_{j}p_j$.
\end{quote}
The above example shows that quasilinear tenants do not satisfy even this weaker assumption: the quasilinear agent prefers the living room although it is the most expensive one.

This raises the question of whether or not the above extensions are valid for quasilinear tenants.
The next section answers this question in the affirmative, by describing a new proof to the existence of rental harmony with quasilinear tenants --- a proof using Sperner's lemma.

\section{Quasilinear Tenants and Sperner's Lemma}
\label{sec:quasilinear}
To handle quasilinear tenants, we just need to change the interpretation of the points in the $n$-simplex: Associate each point $\mathbf{x}$ in $\Delta^{n-1}$
with a price-vector
$\mathbf{p}$ such that, for all $j\in \allrooms$,
\begin{align}
p_j = 1 / x_j.
\end{align}
Then, the boundary points correspond to price-vectors in which some rooms cost infinity. 
This means that each price is an element of the (positive) extended real number line $\mathbb{R_+}\cup\{\infty\}$, rather than an element of $\mathbb{R_+}$.
Note that, since $\sum_{j=1}^n x_j = 1$ and $\forall j: x_j\geq 0$, all prices are positive and at least one price is finite.

As in Su's proof, each vertex of the triangulation is labeled with a best room of the vertex owner. 
A quasilinear agent always prefers a room with a finite price to a room with an infinite price.
Hence, the labeling satisfies Sperner's boundary condition: the label on a boundary vertex is always the index of room $j$ for which $x_j>0$. By Sperner's lemma, the labeling admits a fully-labeled sub-simplex.
Continuity of preferences is preserved too. Hence there exist an envy-free allocation with some price-vector $\mathbf{p}$. 

One problem remains: the sum of prices in $\mathbf{p}$ may be unequal to $R$.
However, with quasilinear agents this is easy to solve.
First, note that all prices in $\mathbf{p}$ are finite --- otherwise no allocation would have been envy-free. 
Let $S := \sum_{j=1}^n p_j$. Let $\mathbf{q}$ be a new price-vector defined by: 
$q_j := p_j + (R-S)/n$.
For a quasilinear agent, adding a fixed amount to the price of each room does not change the relative preference-ordering between the rooms. Hence, the same allocation is envy-free with price-vector $\mathbf{q}$, and the sum of prices is $S+(R-S)=R$.

Note that adding $(R-S)/n$ to all prices might make some prices negative. This means that some tenants are paid to live in their room.%
\footnote{
Negative prices may make sense in some situations. For example, if one of the rooms requires constant maintenance in order to prevent nuisances to the other rooms, then the tenants might agree to pay anyone who will take this room and do the maintenance job. 
}
This is inevitable: in some situations with quasilinear tenants, all envy-free allocations have negative prices.
For example, suppose $n=2$, tenant 1 values the bedroom at $150$ and the basement at $0$, tenant 2 values the bedroom at $140$ and the basement at $10$, and $R=100$. Then in any envy-free allocation 
tenant 1 gets the bedroom and tenant 2 gets the basement, and to avoid envy the difference in prices must be at least $130$. Since the sum of prices is $100$, the price of the basement must be at most $-15$.%
\footnote{
\citet{Brams2001Competitive} show an example with $n=4$ agents and rooms, in which for each agent, the sum of values of all rooms equals $R$, and still there are negative prices in any envy-free allocation.
}

\section{Non-linear Tenants}
\label{sec:nonsatiable}
\citet{svensson1983large} and
\citet{Alkan1991Fair}
generalized the quasilinear tenants model by assuming that each agent $i$ has a transitive preference-relation $\succeq_i$ on (room,price) pairs. For each price-vector $\price$, the best rooms of agent $i$ are the rooms $j$ for which the pairs $(j,p_j)$ are maximal (by $\succeq_i$).
They assume that the preference-relation is continuous and monotonic in the price, i.e., $(j,p)\succeq_i (j,p+\delta)$ whenever $\delta\geq 0$. Without further assumptions, an envy-free allocation might not exist. For example, if $(1,p)\succ_i (2,q)$ for all agents $i$ and prices $p,q$, then the agent who receives room $2$ always envies the agent who receives room $1$. Therefore they make assumptions whose thrust is that every agent can be convinced to select any room, if its price is sufficiently low relative to the other rooms.
The following assumption is made by 
\citet{svensson1983large} before Theorem 1, and by \citet{Alkan1991Fair} at their Theorem 2 proof.%
\footnote{
The term ``Archimedean'' was invented by Rodrigo Velez. 
A similar assumption, called ``Assumption A1'', is presented in a recent survey paper by 
\citet{velez2018equitable}.
}

\begin{quote}
\textbf{Archimedean Tenants assumption:}
There exists a number $T\geq R$ such that an agent always prefers a free room to a room that costs $T$.
Formally, for any agent $i$ and any two rooms $j,j'$:
$(j,0)\succeq_i (j',T)$.
\end{quote}
Note that, if the assumption is satisfied with some $T$, that it is satisfied with any $T' > T$ by the price-monotonicity. Therefore, the assumption $T\geq R$ is without loss of generality.

The Archimedean Tenants assumption is more general than the Quasilinear Tenants assumption: every quasilinear tenant is Archimedean with any $T\geq  \max_{i\in\allagents} \left(\max_j v_{i,j} - \min_j v_{i,j} \right) = $ the largest value-difference between two rooms.
In particular, if the tenant assigns a non-negative value to each room, and the sum of all values is $R$, then $T=R$ will do.

However, it is still not sufficiently general to handle Miserly Tenants. For example, the tenant in the example at the end of Section \ref{sec:intro} cannot be represented by a preference-relation on room-rent pairs, since the preference between the first two rooms depends on the rents of the other rooms.

To handle such \emph{externalities} in the preferences, \citet{velez2016fairness} presents the \emph{compensation assumption} (Definition 2). The following is a slightly simplified version of his assumption, adapted to the language of rooms and rent.
\begin{quote}
\textbf{Compensable Tenants assumption:}
There exists $T\geq R$ such that, 
if there is a room which costs at most $0$ and the most expensive room costs $T$,
then each agent prefers a room that costs less than $T$.
Formally, 
given a price-vector $\price$ in which $\min_j p_j \leq 0$ and $\max_j p_{j} = T$, every agent has a best room $j^*$ with 
$p_{j^*}  < T$.
\end{quote}
%Again we can assume without loss of generality that $T\geq R$.

The Compensable Tenants assumption is more general than all previous assumptions:
\begin{itemize}
\item Every Archimedean tenant (with some $T$) is compensable with the same $T$.
Suppose $p_j \leq 0$ for some $j\in\allrooms$,
and $p_{j'} = T$ for the most expensive room $j'\in\allrooms$.
Every agent weakly prefers $(j,p_j)$ to $(j,0)$ by price-monotonicity and
$(j,0)$ to $(j',T)$ by the Archimedean assumption. Hence each agent has a best room which costs less than $T$.
\item Every (weakly) miserly tenant is compensable with any $T>0$.
Again suppose $p_j \leq 0$ for some $j\in\allrooms$,
and $p_{j'} = T$ for the most expensive room $j'\in\allrooms$.
By the weak miserly tenants assumption, every agent has a best room which is not the most expensive, so it costs less than $T$.
\end{itemize}
Fortunately, even this most general case can be handled by Sperner's lemma: all that is needed is a different interpretation of the points in the $n$-simplex.
\begin{theorem}
\label{thm:compensable}
An envy-free allocation among compensable tenants always exists.
\end{theorem}
\begin{proof}
Associate each point $\mathbf{x}$ in $\Delta^{n-1}$
with a price-vector
$\mathbf{p}$ such that, for all $j\in \allrooms$,
\begin{align}
\label{eq:price}
p_j = T - (T n-R)x_j.
\end{align}
In particular, when $T=R$ this gives $p_j=R\cdot(1-(n-1)x)$.
%For example, with $n=3$ and $T=R$ we get $p_j = R\cdot (1 - 2x_j)$.
Note that $\sum_{j\in\allrooms}p_j = T n-(T n-R) = R$.
If $x_{j'}=0$ for some $j\in\allrooms$, then $p_{j'}=T$.
Moreover, when $x_{j'}=0$, there is at least one other room $j$ for which $x_{j} \geq 1/(n-1)$, which implies $p_{j} \leq (R-T)/(n-1) \leq 0$
(here the assumption $T\geq R$ is used).

The Compensable Tenants assumption then implies that for each agent there is
a best room $j^*\in\allrooms$ for which $p_{j^*}<T$, which implies $x_{j^*}>0$. 
Hence, the agents' labelings satisfy Sperner's boundary condition and an envy-free allocation exists.
\end{proof}

\begin{remark*}
In general, the prices in the envy-free allocation of Theorem \ref{thm:compensable} might be negative. As said in Section \ref{sec:quasilinear}, this may be inevitable since the theorem covers quasilinear tenants.

However, if all tenants are miserly (in addition to being compensable), then the prices in any envy-free allocation must be non-negative,
since otherwise the miserly tenant who gets the most expensive room envies the tenant who gets the room with the negative price.

Thus, the proof combines advantages of  previous proofs: it works both with and without the miserly tenants assumption, and with this assumption it guarantees non-negative prices.
\end{remark*}

The proof of rental harmony existence using Sperner's lemma has two advantages over previous proofs regarding quasilinear agents. 

First, it is arguably simpler. Due to the discrete nature of Sperner's lemma, and thanks to the beautiful and simple proofs available for it \citep{Su1999Rental},
it is easily understood even by students with little background in mathematics.%
\footnote{
This observation is based on the author's experience teaching fair division to computer programmers.
}

Second, the new proof enables to extend the results on rental harmony existence in all settings listed in Section \ref{sec:miserly} (roommates, secretive agents, extra agents or multiple houses)
from miserly tenants to compensable tenants. In particular, all these results hold for quasilinear tenants. These existence results, as far as we know, were not known before.
The following section presents these results briefly.

\section{Extensions}
We will use a generalization of Sperner's lemma, which was proved recently by \citet{meunier2019multilabeled}.
It considers a triangulation of the standard simplex $\Delta^{m-1}$,
which is labeled with labels from the set $\{1,\ldots,m\}$ by $n$ different agents.
For each sub-simplex $\sigma$ of the triangulation, denote by $G(\sigma)$ the bipartite graph with the $n$ agents on one side, the $m$ labels on the other side, and there is an edge between an agent and a label iff the agent uses that label anywhere on that simplex.

Let $G(\sigma)$ be the bipartite graph, and let $w$ be a function that assigns a non-negative weight to each edge of $G(\sigma)$, such that the sum of all weights is $1$.
For each agent $i\range{1}{n}$, let $b(\sigma,w)_i$ be the sum of weights of the edges adjacent to $i$ in $G(\sigma)$.
Similarly, for each label $j\range{1}{m}$, let $a(\sigma,w)_j$ be the sum of weights of the edges adjacent to $j$ in $G(\sigma)$.
Note that the vector $\mathbf{b}(\sigma,w)$ is a point in $\Delta^{n-1}$ and the vector $\mathbf{a}(\sigma,w)$ is a point in $\Delta^{m-1}$.
An illustration is shown in Figure  \ref{fig:weight-vectors} (left).

\begin{figure}
\includegraphics[height=4cm]{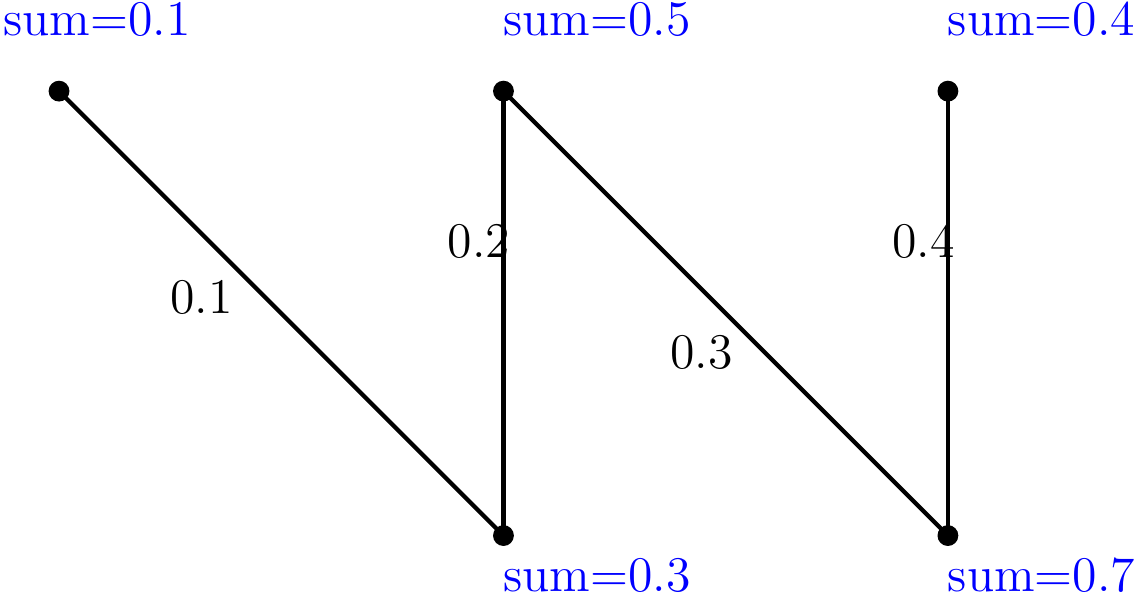}
\hskip 3cm
\includegraphics[height=4cm]{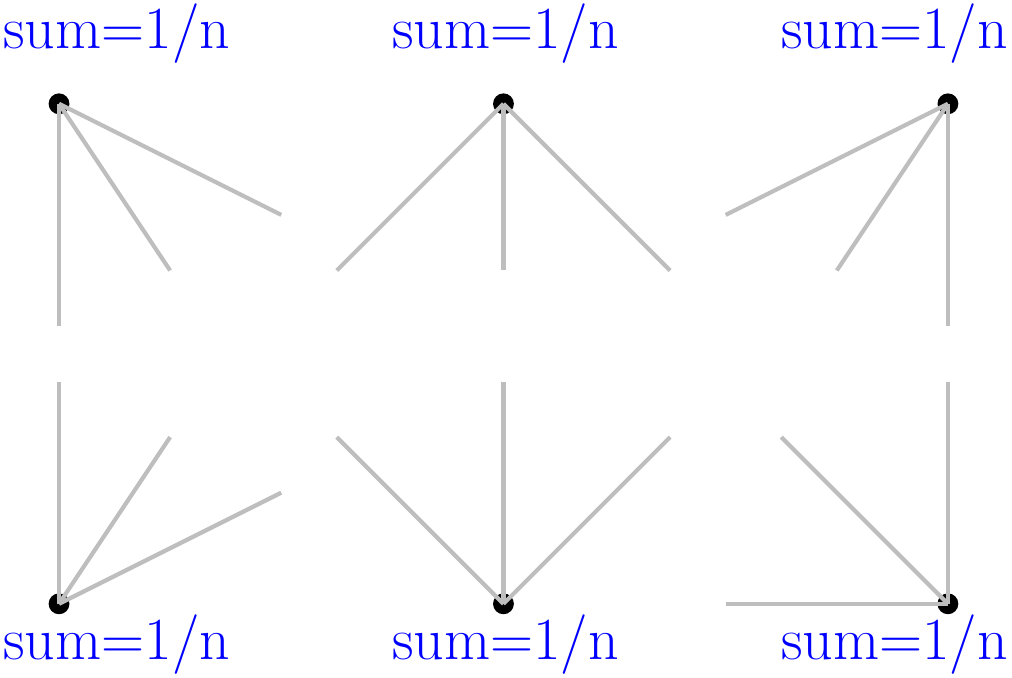}
\caption{
\label{fig:weight-vectors}
\textbf{Left:}
A possible bipartite graph $G(\sigma)$ in a setting with $n=3$ agents and $m=2$ labels.
This graph corresponds to a sub-simplex $\sigma \in\Delta^{2-1}$ in which agent 1 labels both vertices by 1, agent 2 labels one vertex by 1 and one vertex by 2, and agent 3 labels both vertices by 2. The weights on the edges show a possible weight-function $w$. Here $\mathbf{b}(\sigma,w) = [0.1, 0.5, 0.4]$ and $\mathbf{a}(\sigma,w) = [0.3, 0.7]$.
\\
\textbf{Right:} A possible bipartite graph $G(\price)$ in which $\mathbf{a}(w,\price) = \mathbf{b}(w,\price) = (1/n,\ldots,1/n)$.
}
\end{figure}

The following is proved by \citet{meunier2019multilabeled} as a crucial step before the proof of their Theorem 2.3.2.

\begin{lemma}
[\citet{meunier2019multilabeled}]
\label{thm:ms}
Let $\mathbf{a}_0$ be an arbitrary vector in $\Delta^{m-1}$ and 
$\mathbf{b}_0$ an arbitrary vector in $\Delta^{n-1}$.
If $n$ agents label the vertices of a triangulation of $\Delta^{m-1}$ by labels from $\{1,\ldots,m\}$, and all labelings satisfy Sperner's boundary condition, then there exists a sub-simplex $\sigma$ and a weight-function $w$ on $G(\sigma)$ such that 
\begin{align*}
\mathbf{b}(\sigma,w) = \mathbf{b}_0
&&
\mathbf{a}(\sigma,w) = \mathbf{a}_0.
\end{align*}
\end{lemma}

We will apply this lemma to labelings generated as in Theorem \ref{thm:compensable}, with $n$ agents and $m$ rooms. 
Each agent $i$ labels each vertex $\mathbf{x}$ with the index of a best room of $i$ given the price-vector $\mathbf{p}$ calculated from $\mathbf{x}$ by equation \eqref{eq:price}.
As explained in the proof of Theorem \ref{thm:compensable}, all these labelings satisfy Sperner's boundary condition.
For each price-vector $\price\in\mathbb{R}^m$, let $G(\price)$ be the bipartite graph defined like $G(\sigma)$ above: the vertices are the $n$ agents and the $m$ rooms, and each agent is adjacent to all his best room/s given the price $\price$.
By the standard continuity argument, we get the following corollary of Lemma \ref{thm:ms}.
\begin{corollary}
\label{cor:ms}
Suppose there are $m$ rooms and $n$ compensable agents.
Let $\mathbf{a}_0$ be an arbitrary vector in $\Delta^{m-1}$ and 
$\mathbf{b}_0$ be an arbitrary vector in $\Delta^{n-1}$.
There exists a price-vector $\price\in\mathbb{R}^m$ and a weight-function $w$ on $G_\mathbf{p}$ such that 
\begin{align*}
\mathbf{b}(\mathbf{p},w) = \mathbf{b}_0
&&
\mathbf{a}(\mathbf{p},w) = \mathbf{a}_0.
\end{align*}
\end{corollary}

To illustrate the usefulness of Corollary \ref{cor:ms}, let us use it to re-prove the existence of an envy-free room allocation in the standard setting in which $m=n$.
Apply Corollary \ref{cor:ms} with $\mathbf{a}_0 = \mathbf{b}_0 = (1/n,\ldots,1/n)$.
It implies the existence of a price-vector $\price$ and a weight function $w$ on $G(\price)$, which look as in Figure \ref{fig:weight-vectors} (right).
All weights are weakly-positive, and the sum of weights near each vertex is $1/n$. 
This implies that each agent is adjacent to at least one room; every two agents are adjacent together to at least two rooms; and so on. In general, for each subset of $k$ agents, the total weight near their vertices is $k/n$, so they must be adjacent to at least $k$ different rooms. 
This means that the graph $G(\price)$ satisfies the conditions to Hall's marriage theorem \citep{hall1935representatives}. Hence, there exists a perfect matching of agents to rooms, where each agent is matched to one of his/her best rooms given $\price$.

Next, consider the problem of \emph{rental harmony with roommates}.
There are $|\allrooms|=m$ rooms with $m\leq n$,
and  each room $j\in\allrooms$ has a fixed capacity $c_j$, where $\sum_{j\in\allrooms}c_j = n$.
The goal is to assign to each room $j$ a subset $c_j$ of tenants and a price $p_j$ (which is shared equally among the $c_j$ tenants) such that no tenant is envious.
\citet{Azrieli2014Rental} proved the existence of an envy-free allocation assuming all tenants are miserly. 
Later, \citet{ghodsi2018rent} proved the same assuming all tenants are quasilinear.%
\footnote{
They used a reduction to the standard setting: For each room $j$, construct $c_j$ sub-rooms with capacity $1$, and let each agent value all sub-rooms of room $j$ by $v_j$.
Note that this reduction cannot be used for miserly tenants, since it does not preserve the Miserly Tenants assumption. 
For example, suppose the living-room has capacity $2$ and the basement has capacity $1$.
Suppose the two sub-rooms of the living-room are priced at $0$ and $200$ and the basement is priced at $200$. Then, the living-room and the basement have the same positive price, so a tenant who prefers the basement satisfies the Miserly Tenants assumption in the original problem, but not in the reduced problem.
}
The following theorem generalizes both results.
\begin{theorem}
\label{thm:compensable-roommates}
When all agents are compensable, an envy-free allocation with roommates exists.
\end{theorem}
\begin{proof}
Apply Corollary \ref{cor:ms} with $\mathbf{b}_0 = (1/n,\ldots,1/n)$ and $\mathbf{a}_0 = (c_1/n,\ldots,c_m/n)$.
It implies the existence of a price-vector $\price$ such that, in the bipartite graph $G(\price)$, 
each room $j$ is adjacent to at least $c_j$ agents.
Moreover, every two rooms $j_1,j_2$ are adjacent to at least $c_{j_1}+c_{j_2}$ agents, and so on.
By a straightforward generalization of Hall's marriage theorem (see \citet{Azrieli2014Rental}, Theorem 3 in Appendix B), the graph 
$G(\price)$ admits a one-to-many matching in which each room $j$ is matched to exactly $c_j$ agents.
This corresponds to an envy-free room allocation with roommates.
\end{proof}

Next, consider the problem of \emph{rental harmony with a secretive agent}. There are $n\geq 2$ rooms and $n$ agents, but only $n-1$ agents are present. They need to decide on a price-vector $\price\in\mathbb{R}^n$ such that, when the $n$-th agent comes and picks a room, the remaining agents can allocate the remaining rooms among them (without changing the prices) such that there is no envy.
\citet{asada2018fair} and \citet{frick2019achieving} proved that this is possible with miserly tenants. 
The paper of 
\citet{velez2016fairness} implies a slightly weaker result for compensable tenants: an envy-free allocation exists even when only $n-1$ agents are compensable (but the preferences of the $n$-th agent must still be known).
The following theorem generalizes both results.
\begin{theorem}
\label{thm:compensable-secretive}
When at least $n-1$ agents are compensable, an envy-free allocation can be found even before the preferences of the $n$-th agent are known.
\end{theorem}
\begin{proof}
Apply Corollary \ref{cor:ms} with the $n-1$ present agents and the $n$ rooms.
Let $\mathbf{b}_0 = ({1\over n-1},\ldots,{1\over n-1})$ and $\mathbf{a}_0 = ({1\over n},\ldots,{1\over n})$.
It implies the existence of a price-vector $\price$ 
and a weight-function on the bipartite graph $G(\price)$, 
such that the sum of weights near each present agent is ${1\over n-1}$, and the sum of weights near each room is ${1\over n}$.
Suppose that the $n$-th agent now comes and picks a room.
There are $n-1$ remaining rooms.
The weight near each subset of $k\leq n-1$ rooms is at least ${k\over n}$. This fraction is larger than ${k-1\over n-1}$, so every $k$ rooms are adjacent together to more than $k-1$ agents, which means at least $k$ agents.
Hence, the remaining graph satisfies Hall's marriage condition, and there is a perfect matching between the $n-1$ agents and the $n-1$ remaining rooms. 
\end{proof}

\begin{remark*}
\citet{frick2019achieving} asked whether there always exists an envy-free allocation with both roommates and a secretive agent. The answer is yes, and it can be proved by combining the proofs of Theorems \ref{thm:compensable-roommates} and \ref{thm:compensable-secretive}.
\end{remark*}

Next, consider the problem of \emph{rental harmony with an extra agent}. There are $n$ rooms and $n+1$ agents. They need to decide on a price-vector $\price\in\mathbb{R}^n$ such that, 
when any agent leaves, the remaining agents can allocate the rooms among them (without changing the prices) such that there is no envy.
\begin{theorem}
\label{thm:compensable-extra}
When all agents are compensable, an envy-free allocation with an extra agent exists.
\end{theorem}
\begin{proof}
The proof is very similar to Theorem \ref{thm:compensable-secretive}.
Apply Corollary \ref{cor:ms} with the $n+1$ agents and the $n$ rooms.
Let $\mathbf{b}_0 = (1/(n+1),\ldots,1/(n+1))$ and $\mathbf{a}_0 = (1/n,\ldots,1/n)$.
It implies the existence of a price-vector $\price$ 
and a weight-function on the bipartite graph $G(\price)$, 
such that the sum of weights near each agent is $1/(n+1)$, and the sum of weights near each room is $1/n$.
Suppose that agent $i$ leaves. 
Remove from $G(\price)$ the vertex representing $i$ and all its adjacent edges.
In the remaining graph, the weight near each subset of $k$ rooms is at least $k/n - 1/(n+1) > (k-1)/(n+1)$, so every $k$ rooms are adjacent together to at least $k$ remaining agents.
Hence, the remaining graph satisfies Hall's marriage condition, and there is a perfect matching between the $n$ remaining agents and the $n$ rooms. 
\end{proof}

\iffalse
Finally, consider the problem of \emph{multi-house rental harmony}. 
Suppose there are two houses with $m$ rooms in each, e.g. a bedroom house and an office house. The rooms are rented in pairs --- a ``rent package'' consists of a bedroom, and office, and a part of the total rent.
There are $n\geq m$ tenants, each of whom has a demand-function on room pairs.
\citet{nyman2020fair} proved that, if there are at least $2 m-1$ potential tenants, all of whom then there exists a price-vector and a subset of $m$ tenants, such that each tenant can be given his/her preferred pair of rooms given the price (if there are only $n$ tenants then this may be impossible even with $n=2$).
\fi

\section{Future Work}
The new existence results proved above open up some interesting computational issues. 
With quasilinear tenants, there are fast algorithms for computing an envy-free allocation (see Section \ref{sec:intro}).
Can these algorithms handle more advanced settings such as roommates, secretive agents or extra agents? 
Can these algorithms be extended to compensable tenants?

\ifdefined\ACKS
\section{Acknowledgments}
I am grateful to Guillaume Ch\`eze, Yaron Azrieli, Eran Shmaya, Rodrigo Velez, Fr\'ed\'eric Meunier and Shira Zerbib for their kind and helpful comments.
\fi

%\newpage
\bibliographystyle{apalike}
\bibliography{../erelsegal-halevi}

\end{document}